\documentclass[onecolumn,amsmath,amssymb,10pt,aps]{revtex4-2}
  
\usepackage[utf8]{inputenc}
\usepackage[T1]{fontenc}

\usepackage{amsmath}
\usepackage{amsfonts}
\usepackage{graphicx}
\usepackage{yfonts}
\usepackage{color}
\usepackage[normalem]{ulem}
\usepackage{amsthm}
\usepackage{bm}
\usepackage{bbm}
\usepackage{mathtools}
\usepackage{array}
\usepackage{placeins}
\usepackage{enumitem}
\usepackage{tikz}

\def\<{\langle}
\def\>{\rangle}

\newcommand{\Tr}{\mathrm{Tr}}
\def\oper{{\mathchoice{\rm 1\mskip-4mu l}{\rm 1\mskip-4mu l}
{\rm 1\mskip-4.5mu l}{\rm 1\mskip-5mu l}}}
\DeclareMathAlphabet\mathbfcal{OMS}{cmsy}{b}{n}

\mathchardef\mhyphen="2D 

\newtheorem{Remark}{Remark}
\newtheorem{Proposition}{Proposition}
\newtheorem{Example}{Example}

\begin{document}

\title{Entanglement witnesses from mutually unbiased measurements}

\author{Katarzyna Siudzi\'{n}ska and Dariusz Chru\'{s}ci\'{n}ski}
\affiliation{Institute of Physics, Faculty of Physics, Astronomy and Informatics \\  Nicolaus Copernicus University, ul. Grudzi\k{a}dzka 5/7, 87--100 Toru\'{n}, Poland}

\begin{abstract}
A new family of positive, trace-preserving maps is introduced. It is defined using the mutually unbiased measurements, which generalize the notion of mutual unbiasedness of orthonormal bases. This family allows one to define entanglement witnesses whose indecomposability depends on the characteristics of the associated measurement operators.
\end{abstract}

\flushbottom

\maketitle

\thispagestyle{empty}

\section{Introduction}

Entanglement is an important quantum resource used in quantum information theory, quantum communication, quantum cryptography, and quantum computation \cite{HHHH}. Therefore, distinguishing between separable and entangled states is of utmost importance. A quantum state represented by a density operator $\rho$ on the Hilbert space $\mathcal{H}_1\otimes\mathcal{H}_2$ is separable if and only if it can be decomposed into $\rho = \sum_k p_k \rho^{(1)}_k \otimes \rho^{(2)}_k$, where $\rho^{(1)}_k$ and $\rho^{(2)}_k$ are density operators of two subsystems and $p_k$ is a probability distribution. It turns out that  $\rho$ is separable if and only if $(\oper\otimes\Phi)[\rho]\geq 0$ for any positive map $\Phi$ \cite{Horodeccy}. Any quantum state that violates this condition is therefore entangled. Every entangled state can be detected via an entanglement witness \cite{Terhal1,Terhal2}, which is a block-positive operator with at least one negative eigenvalue. Entanglement witnesses $W$ are related to positive but not completely positive maps $\Phi$ via the Choi-Jamio{\l}kowski isomorphism,
\begin{equation}\label{W}
W=\sum_{i,j=0}^{d-1}|i\>\<j|\otimes\Phi[|i\>\<j|],
\end{equation}
where $|k\>$ is an orthonormal basis in $\mathcal{H}\simeq\mathbb{C}^d$. Any such bipartite operator is Hermitian and positive on separable states ($\< \psi \otimes \phi|W|\psi \otimes \phi\>\geq 0$). Entanglement witnesses define a universal mathematical tool to analyze quantum entangled states. For any entangled state $\rho$, there exists an entanglement witness $W$ (which is not unique) such that $\Tr(W \rho ) <0$  (cf. \cite{Guhne,TOPICAL} for recent reviews). 

There are several proposals for constructing such operators. A special class consists in decomposable witnesses, which can be represented as $W=A+B^\Gamma$, where $A,B \geq 0$ and $B^\Gamma$ denotes a partial transposition. However, such witnesses cannot be used to detect bound entanglement; i.e., entangled states that are PPT (positive under partial transposition). Construction of indecomposable witnesses is notoriously hard. There is one class related to the well-known realignment or computable cross-norm (CCNR) separability criterion \cite{R1,R2,R3}, and covariance matrix criterion \cite{CMM1,CMM2,CMM3}. These were recently generalized in Refs. \cite{ESIC-Darek,ESIC-Darek2}. Another interesting class of indecomposable witnesses is constructed in terms of mutually unbiased bases \cite{Spengler,MUBs,Junu1,Junu2,Durt}. This approach was then generalized to mutually unbiased measurements (MUMs) \cite{Kalev} in Refs. \cite{EW-SIC,EW-2MUB}.

In this paper, we further develop the construction of entanglement witnesses (in particular, indecomposable ones) in terms of MUMs, which generalize the notion of mutually unbiased bases (MUBs) to non-projective operators. Contrary to MUBs, of which there are $d+1$ for power prime dimensions $d$ \cite{Wootters,Ivonovic} and at least three otherwise \cite{MUB-2}, one can always construct the maximal number of $d+1$ MUMs. The applications of mutually unbiased measurements range from entropic uncertainty relations \cite{ChenFei,Rastegin,Rastegin2}, through separability criteria for $d$-dimensional bipartite \cite{ChenMa,Shen,ShenLi} and multipartite \cite{Liu,ChenLi} systems, to $k$-nonseparability detection of multipartite qudit systems \cite{Liu3}, and finally to entanglement detection of bipartite states \cite{MUM_purity}. Also, Li et al. used the MUMs to introduce new positive quantum maps and entanglement witnesses \cite{Li} that generalize the construction from Ref. \cite{MUBs}.

In the following sections, we recall the definition of mutually unbiased measurements as well as the method of their construction. Next, we use a set of $N\leq d+1$ MUMs to introduce a family of positive, trace-preserving maps and the corresponding entanglement witnesses. We illustrate our results with several examples. By using mutually unbiased measurement to construct orthonormal Hermitian bases, we prove that these witnesses do not depend on the parameter $\kappa$ that characterizes the MUM properties. However, it turns out that there exists a relation between indecomposability of witnesses and the optimal value of $\kappa$. Finally, we show how our family of entanglement witnesses is related to a large class of witnesses based on the CCNR separability criterion \cite{YuLiu}.

\section{Mutually unbiased measurements}

In quantum information theory, any measurement is represented by a positive, operator-valued measure (POVM) $\{E_\alpha\,|\,E_\alpha\geq 0,\,\sum_\alpha E_\alpha=\mathbb{I}_d\}$. The probability of obtaining the outcome labeled by $\alpha$ is $p_\alpha=\Tr(E_\alpha\rho)$, where $\rho$ is a density operator. A special class of POVMs consists in measurement operators that are orthogonal projectors. Symmetric projective measurements can be performed using mutually unbiased bases. Recall that orthonormal bases $\{\psi_k^{(\alpha)},k=0,\ldots,d-1\}$ in $\mathbb{C}^d$, numbered by $\alpha=1,\ldots,N$, are mutually unbiased if and only if $|\<\psi_k^{(\alpha)}|\psi_l^{(\beta)}\>|^2=1/d$ for $\alpha\neq\beta$. Now, the corresponding rank-1 projectors $P_k^{(\alpha)}=|\psi_k^{(\alpha)}\>\<\psi_k^{(\alpha)}|$, which satisfy the properties
\begin{equation}\label{MUB}
\begin{split}
&\Tr(P_k^{(\alpha)})=1,\\
&\Tr(P_k^{(\alpha)}P_l^{(\beta)})=\delta_{\alpha\beta}\delta_{kl}+\frac 1d (1-\delta_{\alpha\beta}),
\end{split}
\end{equation}
forms a set of $N$ mutually unbiased projective measurements. Kalev and Gour  generalized this notion to non-projective measurement operators \cite{Kalev}. Indeed, the measurements $\{P_k^{(\alpha)}|P_k^{(\alpha)}\geq 0,\sum_{k=0}^{d-1}P_k^{(\alpha)}=\mathbb{I}_d\}$ are mutually unbiased if and only if
\begin{equation}\label{MUM}
\begin{split}
&\Tr(P_k^{(\alpha)})=1,\\
&\Tr(P_k^{(\alpha)}P_l^{(\beta)})=\frac 1d +\frac{d\kappa-1}{d-1}\delta_{\alpha\beta}\left(\delta_{kl}
-\frac 1d \right),
\end{split}
\end{equation}
where $1/d<\kappa\leq 1$. For $\kappa=1$, one reproduces eq. (\ref{MUB}). The maximal number of $d+1$ MUMs forms an informationally complete set and can be constructed using an orthonormal basis $\{\mathbb{I}_d/\sqrt{d},G_{\alpha,k}\}$ of traceless Hermitian operators $G_{\alpha,k}$. The relation between $G_{\alpha,k}$ and $P_k^{(\alpha)}$ is given by the formula
\begin{equation}\label{mum}
P_k^{(\alpha)}=\frac 1d \mathbb{I}_d+tF_k^{(\alpha)},
\end{equation}
where
\begin{equation}\label{f}
F_k^{(\alpha)}=\left\{\begin{aligned}
\sum_{l=1}^{d-1}G_{\alpha,l}-\sqrt{d}(\sqrt{d}+1)G_{\alpha,k},&\qquad k\neq 0,\\
(\sqrt{d}+1)\sum_{l=1}^{d-1}G_{\alpha,l},&\qquad k=0.
\end{aligned}\right.
\end{equation}
The parameter $t$ relates to $\kappa$ via
\begin{equation}\label{t}
\kappa=\frac 1d +(d-1)t^2(1+\sqrt{d})^2,
\end{equation}
and it is chosen in such a way that $P_k^{(\alpha)}\geq 0$. The optimal value $\kappa_{\rm opt}$ is the highest possible value of $\kappa$ for which the condition $P_k^{(\alpha)}\geq 0$ holds. Notably, $\kappa_{\rm opt}$ depends on the choice of the operator basis $G_{\alpha,k}$. For example, $\kappa_{\rm opt}=\frac{d+2}{d^2}$ for the Gell-Mann matrices (see Appendix A), and $\kappa_{\rm opt}=1$ for the basis that gives rise to the mutually unbiased basis.

\section{Positive maps and entanglement witnesses}

Let us consider the trace-preserving map
\begin{equation}\label{PTP}
\Phi=\frac{1}{d\kappa-1}\left[(d\kappa-1-N+2L)\Phi_0+\sum_{\alpha=L+1}^N\Phi_\alpha
-\sum_{\alpha=1}^L\Phi_\alpha\right],
\end{equation}
where $\Phi_0[X]=\mathbb{I}_d\Tr(X)/d$ is the completely depolarizing channel and
\begin{equation}
\Phi_\alpha[X]=\sum_{k,l=0}^{d-1}\mathcal{O}^{(\alpha)}_{kl}P_k^{(\alpha)}\Tr(P_l^{(\alpha)}X).
\end{equation}
The maps $\Phi_\alpha$ are constructed from mutually unbiased measurements $P_k^{(\alpha)}$ and orthogonal rotations $\mathcal{O}^{(\alpha)}$ that preserve the  vector $\mathbf{n}_\ast=(1,\ldots,1)/\sqrt{d}$. Note that the greater the value of $L$ (i.e., the more we subtract), the higher the coefficient that stands before the identity operator $\mathbb{I}_d$.

\begin{Proposition}
The trace-preserving map defined by eq. (\ref{PTP}) is positive.
\end{Proposition}

\begin{proof}
Take an arbitrary rank-1 projector $P$. We prove that
\begin{equation}\label{pos}
\Tr(\Phi[P])^2\leq\frac{1}{d-1},
\end{equation}
which is a sufficient positivity condition for $\Phi$ \cite{MUBs}.
For simplicity, we consider the map
\begin{equation}
\widetilde{\Phi}=a\Phi_0+\sum_{\alpha=L+1}^N\Phi_\alpha-\sum_{\alpha=1}^L\Phi_\alpha,
\end{equation}
where $a=d\kappa-1-N+2L$ and $\widetilde{\Phi}=(d\kappa-1)\Phi$. Now, let us calculate
\begin{equation}\label{PTP1}
\begin{split}
\Tr(\widetilde{\Phi}[P])^2=\Tr\Bigg\{&a^2\Phi_0[P]^2+\sum_{\alpha,\beta=1}^L
\Phi_\alpha[P]\Phi_\beta[P]+\sum_{\alpha,\beta=L+1}^N\Phi_\alpha[P]\Phi_\beta[P]
+2a\sum_{\alpha=L+1}^N\Phi_0[P]\Phi_\alpha[P]\\&-2a\sum_{\alpha=1}^L\Phi_0[P]\Phi_\alpha[P]
-2\sum_{\alpha=L+1}^N\sum_{\beta=1}^L\Phi_\alpha[P]\Phi_\beta[P]\Bigg\}.
\end{split}
\end{equation}
Observe that the subsequent terms can be simplified as follows,
\begin{equation}
\Tr(\Phi_0[P]^2)=\Tr(\Phi_0[P]\Phi_\alpha[P])=\Tr(\Phi_\alpha[P]\Phi_\beta[P])=\frac 1d,\quad \alpha\neq\beta,
\end{equation}
and
\begin{equation}
\Tr(\Phi_\alpha[P]^2)=\frac{1-\kappa}{d-1}+\frac{d\kappa-1}{d-1}\sum_{m=0}^{d-1}
\left[\Tr(P_m^{(\alpha)}P)\right]^2,
\end{equation}
where we used the trace properties of MUMs from eq. (\ref{MUM}), as well as the properties of the orthogonal rotation matrices,
\begin{equation}
\sum_{k=0}^{d-1}\mathcal{O}^{(\alpha)}_{kl}=\sum_{l=0}^{d-1}\mathcal{O}^{(\alpha)}_{kl}=1,\qquad
\sum_{k=0}^{d-1}\mathcal{O}^{(\alpha)}_{kl}\mathcal{O}^{(\alpha)}_{km}=\delta_{lm}.
\end{equation}
Hence, eq. (\ref{PTP1}) reduces to
\begin{equation}\label{PTP2}
\begin{split}
\Tr(\widetilde{\Phi}[P])^2=&\,\frac 1d \Big[a^2+(N-L)(N-L-1)+L(L-1)+2a(N-L)-2aL
-2L(N-L)\Big]\\&+\frac{1-\kappa}{d-1}N+\frac{d\kappa-1}{d-1}\sum_{\alpha=1}^N\sum_{m=0}^{d-1}
\left[\Tr(P_m^{(\alpha)}P)\right]^2,\\
&=\frac 1d \Big[(a+N-2L)^2-N\Big]+\frac{1-\kappa}{d-1}N+\frac{d\kappa-1}{d-1}\sum_{\alpha=1}^N\sum_{m=0}^{d-1}
\left[\Tr(P_m^{(\alpha)}P)\right]^2
\end{split}
\end{equation}
The mutually unbiased measurements satisfy the following property \cite{Rastegin,ChenFei},
\begin{equation}
\sum_{\alpha=1}^N\sum_{k=0}^{d-1}\left[\Tr\left(P_k^{(\alpha)}P\right)\right]^2
\leq \frac{N-1}{d}+\kappa.
\end{equation}
Applying this inequality to eq. (\ref{PTP2}), together with the definition of $a$, results in
\begin{equation}
\begin{split}
\Tr(\widetilde{\Phi}[P])^2\leq\frac{(d\kappa-1)^2-N}{d} +\frac{1-\kappa}{d-1}N+\frac{d\kappa-1}{d-1}\left(\frac{N-1}{d}+\kappa\right)=
\frac{(d\kappa-1)^2}{d-1},
\end{split}
\end{equation}
which finally proves that condition (\ref{pos}) holds.
\end{proof}

\begin{Remark}
In the proof to Proposition 1, out of all the defining properties of MUMs, the positivity condition $P_k^{(\alpha)}\geq 0$ is the only one that is never used. Hence, one can take $P_k^{(\alpha)}\ngeq 0$ to construct $\Phi$ using eq. (\ref{PTP}), and this map is still positive. In other words, any operators $P_k^{(\alpha)}$ that sum up to the identity and satisfy eq. (\ref{MUM}) for an arbitrary real parameter $\kappa$ give rise to a positive, trace-preserving map $\Phi$.
\end{Remark}

Note that the map $\Phi$ generalizes several positive maps already known in the literature:
\begin{itemize}
\item when no inversions are present ($L=N$), one recovers the map considered in ref. \cite{Li};
\item for $L=N$ and $\kappa=1$, $\Phi$ reduces to the map constructed from MUBs \cite{MUBs};
\item if $L=N=d+1$ and there are no rotations ($\mathcal{O}^{(\alpha)}=\mathbb{I}_d$), one arrives at the maps of the type analyzed in ref. \cite{MUM_GPC};
\item if $L=N=d+1$, $\mathcal{O}^{(\alpha)}=\mathbb{I}_d$, and $\kappa=1$, we obtain the generalized Pauli map \cite{ICQC}.
\end{itemize}

Positive maps find important applications in the theory of quantum entanglement, where they are used to detect entangled (non-separable) states.
From definition in eq. (\ref{W}), we find that the entanglement witness corresponding to the positive map $\Phi$ reads
\begin{equation}\label{W2}
W=\frac{d\kappa-1-N+2L}{d}\mathbb{I}_{d^2}
+\sum_{\alpha=L+1}^NH_\alpha-\sum_{\alpha=1}^LH_\alpha,
\end{equation}
where
\begin{equation}
H_\alpha=\sum_{k,l=0}^{d-1}\mathcal{O}_{kl}^{(\alpha)}
\overline{P}_l^{(\alpha)}\otimes P_k^{(\alpha)}.
\end{equation}
Recall that in any dimension $d$ one can always construct the maximal set of $d+1$ MUMs using an orthonormal basis $\{\mathbb{I}_d/\sqrt{d},G_{\alpha,k}\}$ of traceless Hermitian operators $G_{\alpha,k}$. Hence, whenever one knows the full set of $d+1$ mutually unbiased measurements, the entanglement witness is equivalently given by
\begin{equation}\label{WW}
\widetilde{W}=\frac{d(d-1)(\sqrt{d}+1)^2}{d\kappa-1}W=(d-1)(\sqrt{d}+1)^2
\mathbb{I}_{d^2}+\sum_{\alpha=L+1}^NJ_\alpha-\sum_{\alpha=1}^LJ_\alpha,
\end{equation}
with
\begin{equation}\label{J}
J_\alpha=\sum_{k,l=0}^{d-1}\mathcal{O}_{kl}^{(\alpha)}
\overline{F}_l^{(\alpha)}\otimes F_k^{(\alpha)}.
\end{equation}
In the above equation, we used the one-to-one correspondence between $P_k^{(\alpha)}$ and $G_{\alpha,k}$ found by Kalev and Gour \cite{Kalev} to reverse engineer the Hermitian basis from a known complete set of MUMs.
Observe that there is now no dependence of the witness $\widetilde{W}$ on the parameter $\kappa$. This is due to the fact that $\kappa$ characterizes mutually unbiased measurements and not operator bases.

Now, let us propose several examples of entanglement witnesses that fall into the category established by $\widetilde{W}$.

\begin{Example}
First, let us take the maximal values for $N=L=d+1$. Also, assume that there are no rotations, so that $\mathcal{O}^{(\alpha)}=\mathbb{I}_d$ for $\alpha=1,\ldots,d+1$. Finally, fix the operator basis $G_{\alpha,k}$ to be the Gell-Mann matrices (see Appendix A). In this case, we have
\begin{equation}
\sum_{\alpha=1}^dJ_\alpha=d(\sqrt{d}+1)^2\left[\sum_{k,m=0}^{d-1}|k\>\<m|\otimes|k\>\<m|
-\sum_{k=0}^{d-1}|k\>\<k|\otimes|k\>\<k|\right],
\end{equation}
\begin{equation}
J_{d+1}=(\sqrt{d}+1)^2\left[d\sum_{k=0}^{d-1}|k\>\<k|\otimes|k\>\<k|-\mathbb{I}_{d^2}\right].
\end{equation}
This allows us to write the formula for $\widetilde{W}$ from eq. (\ref{WW}) in the form
\begin{equation}
\widetilde{W}=d(\sqrt{d}+1)^2\left[\mathbb{I}_{d^2}-dP_+\right],
\end{equation}
where $P_+=\frac 1d \sum_{i,j=0}^{d-1}|i\>\<j|\otimes|i\>\<j|$ is the maximally entangled state. This is exactly the entanglement witness corresponding to the reduction map \cite{EW_reduction}.
\end{Example}

\begin{Example}\label{Ex2}
Now, consider the dimension $d=3$ and take, as in the previous example, $N=L=4$ as well as $\mathcal{O}^{(\alpha)}=\mathbb{I}_3$ for $\alpha=1,2,3$. However, assume that the final rotation matrix $\mathcal{O}^{(4)}=S_k$ describes a permutation, where
\begin{equation}
S_1=\begin{pmatrix}
0 & 0 & 1 \\
1 & 0 & 0 \\
0 & 1 & 0
\end{pmatrix},\qquad
S_2=\begin{pmatrix}
0 & 1 & 0 \\
0 & 0 & 1 \\
1 & 0 & 0
\end{pmatrix}.
\end{equation}
Denote the witness corresponding to the choice $\mathcal{O}^{(4)}=S_k$ by $\widetilde{W}_k$. For the Gell-Mann matrices, one finds
\begin{equation}
\widetilde{W}_1=6(2+\sqrt{3})\left[\begin{array}{c c c|c c c|c c c}
1 & \cdot & \cdot & \cdot & -1 & \cdot & \cdot & \cdot & -1 \\
\cdot & \cdot & \cdot & \cdot & \cdot & \cdot & \cdot & \cdot & \cdot \\
\cdot & \cdot & 1 & \cdot & \cdot & \cdot & \cdot & \cdot & \cdot \\
\hline
\cdot & \cdot & \cdot & 1 & \cdot & \cdot & \cdot & \cdot & \cdot \\
-1 & \cdot & \cdot & \cdot & 1 & \cdot & \cdot & \cdot & -1 \\
\cdot & \cdot & \cdot & \cdot & \cdot & \cdot & \cdot & \cdot & \cdot \\
\hline
\cdot & \cdot & \cdot & \cdot & \cdot & \cdot & \cdot & \cdot & \cdot \\
\cdot & \cdot & \cdot & \cdot & \cdot & \cdot & \cdot & 1 & \cdot \\
-1 & \cdot & \cdot & \cdot & -1 & \cdot & \cdot & \cdot & 1
\end{array}\right],\qquad
\widetilde{W}_2=6(2+\sqrt{3})\left[\begin{array}{c c c|c c c|c c c}
1 & \cdot & \cdot & \cdot & -1 & \cdot & \cdot & \cdot & -1 \\
\cdot & 1 & \cdot & \cdot & \cdot & \cdot & \cdot & \cdot & \cdot \\
\cdot & \cdot & \cdot & \cdot & \cdot & \cdot & \cdot & \cdot & \cdot \\
\hline
\cdot & \cdot & \cdot & \cdot & \cdot & \cdot & \cdot & \cdot & \cdot \\
-1 & \cdot & \cdot & \cdot & 1 & \cdot & \cdot & \cdot & -1 \\
\cdot & \cdot & \cdot & \cdot & \cdot & 1 & \cdot & \cdot & \cdot \\
\hline
\cdot & \cdot & \cdot & \cdot & \cdot & \cdot & 1 & \cdot & \cdot \\
\cdot & \cdot & \cdot & \cdot & \cdot & \cdot & \cdot & \cdot & \cdot \\
-1 & \cdot & \cdot & \cdot & -1 & \cdot & \cdot & \cdot & 1
\end{array}\right],
\end{equation}
which belong to the Choi-type maps analyzed in ref. \cite{HaKye2}. For clarity, all zeros are represented by dots.
The same witnesses can also be obtained from the mutually unbiased bases \cite{MUBs}.
\end{Example}

Unfortunately, these simple forms of entanglement witnesses constructed from the Gell-Mann matrices are not preserved for $d>3$. However, this can be remedied if one modifies the diagonal matrices in the Gell-Mann basis.

\begin{Example}\label{EX}
In what follows, we generalize Example \ref{Ex2} to an arbitrary finite dimension $d$. Once again, we take the maximal $N=L=d+1$, $\mathcal{O}^{(\alpha)}=\mathbb{I}_d$ for $\alpha=1,\ldots,d$, and the permutation matrix $\mathcal{O}^{(d+1)}=S^{(r)}$, $S^{(r)}|i\>=|i+r\>$. Instead of the Gell-Mann matrices, in the construction of the entanglement witness, we use the operator basis introduced in Appendix B. Due to this change, $J_{d+1}$ can be simplified to
\begin{equation}
J_{d+1}=(\sqrt{d}+1)^2\left[d\sum_{k=0}^{d-1}|k-r\>\<k-r|\otimes|k\>\<k|-\mathbb{I}_{d^2}\right].
\end{equation}
Now, the associated witness $\widetilde{W}$ is given by
\begin{equation}
\widetilde{W}=d(\sqrt{d}+1)^2\left[\mathbb{I}_{d^2}-dR\right],
\end{equation}
where
\begin{equation}
R=\frac 1d \left[\sum_{k,m=0}^{d-1}|k\>\<m|\otimes|k\>\<m|-\sum_{k=0}^{d-1}
\big(|k\>\<k|-|k-r\>\<k-r|\big)\otimes|k\>\<k|\right].
\end{equation}
\end{Example}

Interestingly, even though the parameter $\kappa$ that characterizes mutually unbiased measurements is not present in the formula for $\widetilde{W}$,
the properties of entanglement witnesses depend on the optimal (maximal) value of $\kappa$. Indeed, for the MUMs constructed from the Gell-Mann matrices with $\kappa_{\rm opt}=\frac{d+2}{d^2}$, there are less indecomposable witnesses than for the MUBs, where $\kappa_{\rm opt}=1$.

\begin{Example}
In dimension $d=3$, assuming the maximal value of $N=4$ and no rotations ($\mathcal{O}^{(\alpha)}=\mathbb{I}_3$ for $\alpha=1,\ldots,4$), there are five indecomposable witnesses that can be constructed from the mutually unbiased bases. Some examples for $L=2$ are
\begin{equation}
\widetilde{W}_1=2(2+\sqrt{3})
\left[\begin{array}{c c c|c c c|c c c}
\cdot & \cdot & \cdot & \cdot & 1 & \cdot & \cdot & \cdot & 1 \\
\cdot & 3 & \cdot & \cdot & \cdot & -2 & -2 & \cdot & \cdot \\
\cdot & \cdot & 3 & -2 & \cdot & \cdot & \cdot & -2 & \cdot \\
\hline
\cdot & \cdot & -2 & 3 & \cdot & \cdot & \cdot & -2 & \cdot \\
1 & \cdot & \cdot & \cdot & \cdot & \cdot & \cdot & \cdot & 1 \\
\cdot & -2 & \cdot & \cdot & \cdot & 3 & -2 & \cdot & \cdot \\
\hline
\cdot & -2 & \cdot & \cdot & \cdot & -2 & 3 & \cdot & \cdot \\
\cdot & \cdot & -2 & -2 & \cdot & \cdot & \cdot & 3 & \cdot \\
1 & \cdot & \cdot & \cdot & 1 & \cdot & \cdot & \cdot & \cdot
\end{array}\right],
\qquad
\widetilde{W}_2=2(2+\sqrt{3})
\left[\begin{array}{c c c|c c c|c c c}
4 & \cdot & \cdot & \cdot & -1 & \cdot & \cdot & \cdot & -1 \\
\cdot & 1 & \cdot & \cdot & \cdot & 2 & 2 & \cdot & \cdot \\
\cdot & \cdot & 1 & 2 & \cdot & \cdot & \cdot & 2 & \cdot \\
\hline
\cdot & \cdot & 2 & 1 & \cdot & \cdot & \cdot & 2 & \cdot \\
-1 & \cdot & \cdot & \cdot & 4 & \cdot & \cdot & \cdot & -1 \\
\cdot & 2 & \cdot & \cdot & \cdot & 1 & 2 & \cdot & \cdot \\
\hline
\cdot & 2 & \cdot & \cdot & \cdot & 2 & 1 & \cdot & \cdot \\
\cdot & \cdot & 2 & 2 & \cdot & \cdot & \cdot & 1 & \cdot \\
-1 & \cdot & \cdot & \cdot & -1 & \cdot & \cdot & \cdot & 4
\end{array}\right],
\end{equation}
which can be used to detect the positive partial transpose (PPT) states
\begin{equation}
\rho_1=\frac{1}{24}
\left[\begin{array}{c c c|c c c|c c c}
4 & \cdot & \cdot & \cdot & 1 & \cdot & \cdot & \cdot & 1 \\
\cdot & 2 & \cdot & \cdot & \cdot & 2 & 2 & \cdot & \cdot \\
\cdot & \cdot & 2 & 2 & \cdot & \cdot & \cdot & 2 & \cdot \\
\hline
\cdot & \cdot & 2 & 2 & \cdot & \cdot & \cdot & 2 & \cdot \\
1 & \cdot & \cdot & \cdot & 4 & \cdot & \cdot & \cdot & 1 \\
\cdot & 2 & \cdot & \cdot & \cdot & 2 & 2 & \cdot & \cdot \\
\hline
\cdot & 2 & \cdot & \cdot & \cdot & 2 & 2 & \cdot & \cdot \\
\cdot & \cdot & 2 & 2 & \cdot & \cdot & \cdot & 2 & \cdot \\
1 & \cdot & \cdot & \cdot & 1 & \cdot & \cdot & \cdot & 4
\end{array}\right],\quad
\rho_2=\frac{1}{3(3+\sqrt{3})}
\left[\begin{array}{c c c|c c c|c c c}
\sqrt{3}-1 & \cdot & \cdot & \cdot & \sqrt{3}-1 & \cdot & \cdot & \cdot & \sqrt{3}-1 \\
\cdot & 2 & \cdot & \cdot & \cdot & -1 & -1 & \cdot & \cdot \\
\cdot & \cdot & 2 & -1 & \cdot & \cdot & \cdot & -1 & \cdot \\
\hline
\cdot & \cdot & -1 & 2 & \cdot & \cdot & \cdot & -1 & \cdot \\
\sqrt{3}-1 & \cdot & \cdot & \cdot & \sqrt{3}-1 & \cdot & \cdot & \cdot & \sqrt{3}-1 \\
\cdot & -1 & \cdot & \cdot & \cdot & 2 & -1 & \cdot & \cdot \\
\hline
\cdot & -1 & \cdot & \cdot & \cdot & -1 & 2 & \cdot & \cdot \\
\cdot & \cdot & -1 & -1 & \cdot & \cdot & \cdot & 2 & \cdot \\
\sqrt{3}-1 & \cdot & \cdot & \cdot & \sqrt{3}-1 & \cdot & \cdot & \cdot & \sqrt{3}-1
\end{array}\right],
\end{equation}
respectively.

On the contrary, all witnesses that arise from the Gell-Mann basis are decomposable, including
\begin{equation}
\widetilde{W}_3=6(2+\sqrt{3})
\left[\begin{array}{c c c|c c c|c c c}
\cdot & \cdot & \cdot & \cdot & \cdot & \cdot & \cdot & \cdot & 1 \\
\cdot & 1 & \cdot & -1 & \cdot & \cdot & \cdot & \cdot & \cdot \\
\cdot & \cdot & 1 & \cdot & \cdot & \cdot & \cdot & \cdot & \cdot \\
\hline
\cdot & -1 & \cdot & 1 & \cdot & \cdot & \cdot & \cdot & \cdot \\
\cdot & \cdot & \cdot & \cdot & \cdot & \cdot & \cdot & \cdot & \cdot \\
\cdot & \cdot & \cdot & \cdot & \cdot & 1 & \cdot & 1 & \cdot \\
\hline
\cdot & \cdot & \cdot & \cdot & \cdot & \cdot & 1 & \cdot & \cdot \\
\cdot & \cdot & \cdot & \cdot & \cdot & 1 & \cdot & 1 & \cdot \\
1 & \cdot & \cdot & \cdot & \cdot & \cdot & \cdot & \cdot & \cdot
\end{array}\right],
\qquad
\widetilde{W}_4=2(2+\sqrt{3})
\left[\begin{array}{c c c|c c c|c c c}
4 & \cdot & \cdot & \cdot & -3 & \cdot & \cdot & \cdot & \cdot \\
\cdot & 1 & \cdot & \cdot & \cdot & \cdot & \cdot & \cdot & \cdot \\
\cdot & \cdot & 1 & \cdot & \cdot & \cdot & 3 & \cdot & \cdot \\
\hline
\cdot & \cdot & \cdot & 1 & \cdot & \cdot & \cdot & \cdot & \cdot \\
-3 & \cdot & \cdot & \cdot & 4 & \cdot & \cdot & \cdot & \cdot \\
\cdot & \cdot & \cdot & \cdot & \cdot & 1 & \cdot & 3 & \cdot \\
\hline
\cdot & \cdot & 3 & \cdot & \cdot & \cdot & 1 & \cdot & \cdot \\
\cdot & \cdot & \cdot & \cdot & \cdot & 3 & \cdot & 1 & \cdot \\
\cdot & \cdot & \cdot & \cdot & \cdot & \cdot & \cdot & \cdot & 4
\end{array}\right].
\end{equation}
Recall that a witness $W$ is decomposable if it can be written as $W=A+B^\Gamma$, where $A,B\geq 0$ and $\Gamma$ denotes the partial transposition with respect to the second subsystem. In our example, $\widetilde{W}_3$ can be decomposed into $\widetilde{W}_3=A_3+B_3^\Gamma$ with the positive operators
\begin{equation}
A_3=6(2+\sqrt{3})
\left[\begin{array}{c c c|c c c|c c c}
\cdot & \cdot & \cdot & \cdot & \cdot & \cdot & \cdot & \cdot & \cdot \\
\cdot & 1 & \cdot & -1 & \cdot & \cdot & \cdot & \cdot & \cdot \\
\cdot & \cdot & \cdot & \cdot & \cdot & \cdot & \cdot & \cdot & \cdot \\
\hline
\cdot & -1 & \cdot & 1 & \cdot & \cdot & \cdot & \cdot & \cdot \\
\cdot & \cdot & \cdot & \cdot & \cdot & \cdot & \cdot & \cdot & \cdot \\
\cdot & \cdot & \cdot & \cdot & \cdot & 1 & \cdot & 1 & \cdot \\
\hline
\cdot & \cdot & \cdot & \cdot & \cdot & \cdot & \cdot & \cdot & \cdot \\
\cdot & \cdot & \cdot & \cdot & \cdot & 1 & \cdot & 1 & \cdot \\
\cdot & \cdot & \cdot & \cdot & \cdot & \cdot & \cdot & \cdot & \cdot
\end{array}\right],\qquad
B_3=6(2+\sqrt{3})
\left[\begin{array}{c c c|c c c|c c c}
\cdot & \cdot & \cdot & \cdot & \cdot & \cdot & \cdot & \cdot & \cdot \\
\cdot & \cdot & \cdot & \cdot & \cdot & \cdot & \cdot & \cdot & \cdot \\
\cdot & \cdot & 1 & \cdot & \cdot & \cdot & 1 & \cdot & \cdot \\
\hline
\cdot & \cdot & \cdot & \cdot & \cdot & \cdot & \cdot & \cdot & \cdot \\
\cdot & \cdot & \cdot & \cdot & \cdot & \cdot & \cdot & \cdot & \cdot \\
\cdot & \cdot & \cdot & \cdot & \cdot & \cdot & \cdot & \cdot & \cdot \\
\hline
\cdot & \cdot & 1 & \cdot & \cdot & \cdot & 1 & \cdot & \cdot \\
\cdot & \cdot & \cdot & \cdot & \cdot & \cdot & \cdot & \cdot & \cdot \\
\cdot & \cdot & \cdot & \cdot & \cdot & \cdot & \cdot & \cdot & \cdot
\end{array}\right],
\end{equation}
whereas $\widetilde{W}_4$ is decomposable into $\widetilde{W}_4=A_4+B_4^\Gamma$
with positive
\begin{equation}
A_4=2(2+\sqrt{3})
\left[\begin{array}{c c c|c c c|c c c}
\cdot & \cdot & \cdot & \cdot & \cdot & \cdot & \cdot & \cdot & \cdot \\
\cdot & 1 & \cdot & -1 & \cdot & \cdot & \cdot & \cdot & \cdot \\
\cdot & \cdot & 1 & \cdot & \cdot & \cdot & 1 & \cdot & \cdot \\
\hline
\cdot & -1 & \cdot & 1 & \cdot & \cdot & \cdot & \cdot & \cdot \\
\cdot & \cdot & \cdot & \cdot & \cdot & \cdot & \cdot & \cdot & \cdot \\
\cdot & \cdot & \cdot & \cdot & \cdot & 1 & \cdot & 1 & \cdot \\
\hline
\cdot & \cdot & 1 & \cdot & \cdot & \cdot & 1 & \cdot & \cdot \\
\cdot & \cdot & \cdot & \cdot & \cdot & 1 & \cdot & 1 & \cdot \\
\cdot & \cdot & \cdot & \cdot & \cdot & \cdot & \cdot & \cdot & \cdot
\end{array}\right],\qquad
B_4=2(2+\sqrt{3})
\left[\begin{array}{c c c|c c c|c c c}
4 & \cdot & \cdot & \cdot & -2 & \cdot & \cdot & \cdot & 2 \\
\cdot & \cdot & \cdot & \cdot & \cdot & \cdot & \cdot & \cdot & \cdot \\
\cdot & \cdot & \cdot & \cdot & \cdot & \cdot & \cdot & \cdot & \cdot \\
\hline
\cdot & \cdot & \cdot & \cdot & \cdot & \cdot & \cdot & \cdot & \cdot \\
-2 & \cdot & \cdot & \cdot & 4 & \cdot & \cdot & \cdot & 2 \\
\cdot & \cdot & \cdot & \cdot & \cdot & \cdot & \cdot & \cdot & \cdot \\
\hline
\cdot & \cdot & \cdot & \cdot & \cdot & \cdot & \cdot & \cdot & \cdot \\
\cdot & \cdot & \cdot & \cdot & \cdot & \cdot & \cdot & \cdot & \cdot \\
2 & \cdot & \cdot & \cdot & 2 & \cdot & \cdot & \cdot & 4
\end{array}\right].
\end{equation}
\end{Example}

Now, observe that $J_\alpha$ from eq. (\ref{J}) can be expressed directly through the elements of the operator basis $G_{\alpha,k}$, as the operators $F_k^{(\alpha)}$ depend directly on $G_{\alpha,k}$. Indeed, after writing out $F_k^{(\alpha)}$ using eq. (\ref{f}), one arrives at
\begin{equation}
J_\alpha=\sum_{k,l=1}^{d-1}\mathcal{Q}_{kl}^{(\alpha)}\overline{G}_{\alpha,l}
\otimes G_{\alpha,k}
\end{equation}
with
\begin{equation}
\mathcal{Q}_{kl}^{(\alpha)} = d(\mathcal{O}_{00}^{(\alpha)}-1) + d(\sqrt{d}+1)^2 \mathcal{O}_{kl}^{(\alpha)} - d(\sqrt{d}+1) (\mathcal{O}_{0l}^{(\alpha)}+\mathcal{O}_{k0}^{(\alpha)}).
\end{equation}
In the above formula, $\mathcal{Q}^{(\alpha)}$ are rescaled orthogonal matrices since $\mathcal{Q}^{(\alpha)T}\mathcal{Q}^{(\alpha)}
=\mathcal{Q}^{(\alpha)}\mathcal{Q}^{(\alpha)T}=d^2(\sqrt{d}+1)^4\mathbb{I}_{d-1}$. Assume for now that $N=L=d+1$. Then, the corresponding entanglement witness $\widetilde{W}$ is given by
\begin{equation}
\widetilde{W}=d(\sqrt{d}+1)^2\left[\mathbb{I}_{d^2}- d(\sqrt{d}+1)^2 G_0\otimes G_0-\frac{1}{d(\sqrt{d}+1)^2}\sum_{\alpha=1}^{d+1}J_\alpha\right],
\end{equation}
where $G_0=\mathbb{I}_d/\sqrt{d}$. After a simple relabelling of indices, $(\alpha,k)\longmapsto\mu$, it follows that
\begin{equation}
\widetilde{W}=d(\sqrt{d}+1)^2\left[\mathbb{I}_{d^2}-\sum_{\mu,\nu=0}^{d^2-1}Q_{\mu\nu} G_\mu^T\otimes G_\nu\right]
\end{equation}
with a block-diagonal orthogonal matrix
\begin{equation}
Q=\frac{1}{d(\sqrt{d}+1)^2}
\left[\begin{array}{c c c c c}
d(\sqrt{d}+1)^2 & & & & \\
 & \mathcal{Q}^{(1)} & & & \\
 & & \mathcal{Q}^{(2)} & & \\
 & & & \ddots & \\
 & & & & \mathcal{Q}^{(d+1)}
\end{array}\right].
\end{equation}
Therefore, the entanglement witnesses constructed from $d+1$ MUMs belong to a larger class of witnesses \cite{YuLiu}
\begin{equation}\label{WW2}
W^\prime=\mathbb{I}_{d^2}-\sum_{\mu,\nu=0}^{d^2-1}Q_{\mu\nu} G_\mu^T\otimes G_\nu
\end{equation}
defined with the use of any orthonormal Hermitian bases $G_\mu$ and an orthogonal matrix $Q_{\mu\nu}$. Actually, it is enough that $Q^TQ\leq\mathbb{I}_{d^2}$, so the above consideration is also true for both $N\leq d+1$ and $L\leq d+1$ if one allows for the rotation matrices $\widetilde{\mathcal{O}}^{(\alpha)}$ to change the sign of $\mathbf{n}_\ast$ ($\widetilde{\mathcal{O}}^{(\alpha)}\mathbf{n}_\ast=\pm\mathbf{n}_\ast$).

\section{Conclusions}

In this paper, we constructed a family of positive, trace-preserving maps using $N\leq d+1$ mutually unbiased measurements and orthogonal matrices. We showed that these maps give rise to entanglement witnesses regardless of the value of the parameter $\kappa$, which means that the positivity of mutually unbiased measurements is not required. We provided several interesting examples of witnesses for $d=3$ as well as an arbitrary finite dimension using the construction of MUMs from two different orthonormal Hermitian bases. We also showed that there is a relation between indecomposability of witnesses and the optimal value of $\kappa$. At last, we proved that our construction belongs to the family of witnesses based on the CCNR separability criterion  \cite{YuLiu}, where the Hermitian basis consists in the identity and traceless operators, and the orthogonal matrix is block-diagonal. It would be interesting to provide a multipartite generalization of this construction. 


\section{Acknowledgements}

This paper was supported by the Polish National Science Centre project No. 2018/30/A/ST2/00837. K.S. was also supported by the Foundation for Polish Science (FNP).

\appendix

\section{Gell-Mann matrices}

Let us take the generalized Gell-Mann matrices defined on $\mathcal{H}\simeq\mathbb{C}^d$ via
\begin{align}
&\sigma_{kl}:=\frac{1}{\sqrt{2}}\left(|k\>\<l|+|l\>\<k|\right),\\
&\sigma_{lk}:=\frac{i}{\sqrt{2}}\left(|k\>\<l|-|l\>\<k|\right),\\
&\sigma_{kk}:=\sqrt{\frac{1}{k(k+1)}}\left(\sum_{j=0}^{k-1}|j\>\<j|-k|k\>\<k|\right),
\end{align}
where $0\leq k<l\leq d-1$ and $0\leq k\leq d-1$, respectively.
To construct the mutually unbiased measurements, we group them as in \cite{Kalev},
\begin{align*}
&\{G_{\alpha,k}|k=1,\ldots,d-1\}=\{\sigma_{k,\alpha-1}|k\neq\alpha-1\},\\
&\{G_{d+1,k}|k=1,\ldots,d-1\}=\{\sigma_{kk}|k=1,\ldots,d-1\}.
\end{align*}
Then, one has \cite{MUM_GPC}
\begin{align}
&F_0^{(d+1)}=(\sqrt{d}+1)\sum_{l=1}^{d-1}\sigma_{ll},\label{od}\\
&F_k^{(d+1)}=-\sqrt{d}(\sqrt{d}+1)\sigma_{kk}+\sum_{l=1}^{d-1}\sigma_{ll}
\end{align}
for $k=1,\ldots,d-1$, as well as
\begin{align}
&F_k^{(\alpha)}=-\sqrt{d}(\sqrt{d}+1)\sigma_{k,\alpha-1}+\sum_{l\neq\alpha-1}\sigma_{l,\alpha-1},\\
&F_{\alpha-1}^{(\alpha)}=(\sqrt{d}+1)\sum_{l\neq\alpha-1}\sigma_{l,\alpha-1}\label{do}
\end{align}
for $k\neq\alpha-1$ and $\alpha=1,\ldots,d$.

\section{A New Hermitian basis}

For the purposes of Example \ref{EX}, we introduce a new Hermitian operator basis $\sigma_{kl}^\prime$. Assume that the operators with only off-diagonal elements are the same as the Gell-Mann matrices; that is,
\begin{equation}
\sigma_{kl}^\prime=\sigma_{kl},\qquad k\neq l.
\end{equation}
Now, the diagonal operators are defined as $\sigma_{00}^\prime=\mathbb{I}_d/\sqrt{d}$ and
\begin{equation}
\sigma_{kk}^\prime=\frac{1}{\sqrt{d}(\sqrt{d}+1)}\left(\mathbb{I}_d+\sqrt{d}|0\>\<0|\right)
-|k\>\<k|,\qquad k=1,\ldots,d-1.
\end{equation}
Obviously, $\sigma_{kk}^\prime$ are Hermitian, traceless, and together with $\sigma_{kl}$ they form an operator basis. To check that this basis is indeed orthonormal, it is enough to show that
\begin{equation}
\Tr \sigma_{kk}^\prime\sigma_{ll}^\prime=\frac{2d+2\sqrt{d}}{d(\sqrt{d}+1)^2}+
\delta_{kl}-\frac{d}{\sqrt{2}(\sqrt{d}+1)}=\delta_{kl}
\end{equation}
for $k,l=1,\ldots,d-1$, as well as
\begin{equation}
\Tr \sigma_{kk}^\prime\sigma_{lm}^\prime=0
\end{equation}
for $k=1,\ldots,d-1$ and $l,m=0,\ldots,d-1$, $l\neq m$.

\bibliography{C:/Users/cynda/OneDrive/Fizyka/bibliography}
\bibliographystyle{C:/Users/cynda/OneDrive/Fizyka/beztytulow2}


\end{document}